\documentclass[conference, final, 10pt]{IEEEtran}
\usepackage{times,amsthm}
\usepackage{graphicx}
\usepackage[center,font=footnotesize]{caption}
\usepackage{nonfloat}
\usepackage{url}
\usepackage{comment}
\usepackage{listings}
\usepackage{color}

\newtheorem{theorem}{Theorem}
\newtheorem{lemma}{Lemma}
\newtheorem{corollary}{Corollary}

\begin{document}

\title{Fundamental Results for a Generic Implementation of Barriers using Optical Interconnects\vspace{-6mm}}
\author{\IEEEauthorblockN{Sandeep Chandran, Eldhose Peter, Preeti Ranjan Panda, and Smruti R. Sarangi\\}
\IEEEauthorblockA{Department of Computer Science and Engineering, \\
Indian Institute of Technology Delhi,  \\
Hauz Khas, New Delhi -- 110016, India \\
\{sandeep, eldhose, panda, srsarangi\}@cse.iitd.ac.in}
}
\maketitle

\begin{abstract}
In this report, we report some fundamental results and bounds on the number of messages and
storage required to implement barriers using futuristic on-chip optical and RF networks. 
We prove that it is necessary to maintain a count to at least $N$ (number of threads) in memory, broadcast the barrier id at least once, and if we elect a co-ordinator, we can 
reduce the number of messages by a factor of $O(N)$.

\end{abstract}
	
\section{Introduction}

We consider a generic implementation of a barrier over futuristic interconnects that support efficient
broadcast/multicast
mechanisms between participating threads. Two early works in this area~\cite{prvulovic,binkert} propose 
methods for implementing barriers using RF and optical interconnects respectively. These proposals implement a 
wired-AND kind of logic using a broadcast bus. For example, each core transmits a signal (RF or optical) at 
a predetermined wavelength on a broadcast bus when it reaches a barrier. When a core receives signals from all 
the other cores, it decides to proceed past the barrier. 

Such simplistic proposals that assume every core is interested in entering the barrier, and disallow context 
switches or thread migrations, do not work on actual systems where context switches, thread migrations, 
and unknown participants are the norm. Handling these cases without sacrificing the ultra-fast release 
latencies offered by futuristic interconnects is non-trivial and requires several fundamental changes to 
the design of a barrier implementation. 

Here, we prove several key properties that must hold for any ultra-fast barrier implementation that accounts for 
unknown participants, context-switches and thread migrations.

\section{Proofs of key results}

Let us assume that the total number of threads in a barrier group is $N$, and $M$ is the number of threads that 
are currently swapped in and are running on cores. Let the number of cores be $\mathcal{C}$. We assume that each
core is associated with a {\em barrier} controller.
Let us further assume that a barrier controller can only maintain state for the thread that it is currently 
associated with (is scheduled on its associated core). 

We also assume that each barrier group has unknown participants. This means that a thread does not know the 
identities of other threads in the barrier group, and the cores that they are scheduled on. This is typically the 
case where the OS can freely schedule any thread on any core, and even migrate them. The threads involved in a 
barrier group can sometimes be deduced at compile time; however, sometimes they are dependent on runtime parameters. 
We consider the general case. Additionally, we assume a synchronous algorithm where we have the notion of a round.
Every message is sent at the beginning of a round, and received at the end of a round.

\begin{theorem}
We need to store at least $\lceil log(N-M) \rceil$ bits in memory.
\label{thm:nm}
\end{theorem}

\begin{proof}
The $M$ threads that are swapped in can indicate their barrier status by either sending or diverting
an optical signal (see~\cite{prvulovic,binkert}). Their state is maintained in the barrier controllers. 
The state of the remaining $N-M$ threads cannot be maintained at the barrier controllers because we
assume that the barrier controllers maintain state for only the thread that is scheduled on their associated 
core. We thus need to store the state of the $N-M$ threads in memory, which we are assuming to be the only 
other storage location. For the $N-M$ threads, we are not interested in the barrier status of every thread. 
We are interested in knowing if all the threads that have been swapped out ($N-M$ in number) have reached the 
barrier or not. However, 1 bit is not sufficient. We need to in fact prove that any number less than 
$\lceil log(N-M) \rceil$ bits is not sufficient.

Let us designate $\lceil log(N-M) \rceil$ as $\kappa$. Let us assume that we have $\lambda$ bits, where $\lambda <
\kappa$. In this case, we can support a maximum of $2^\lambda$ states, which is less than $N-M$. Now, let us 
assume that we start with a state where none of the $\kappa$ threads have reached the barrier. Sequentially, 
one of them wakes up reaches the barrier, and then waits (thread gets swapped out). If we want a change in a 
state for every such action we need at least $N-M$ states. $\lambda$ bits are not sufficient. This means that 
we cannot have a state change for every thread in this group of $\kappa$ threads reaching a barrier. 
Let us assume a thread $t$, which reaches a barrier, and we do not change any state subsequently. 
We will have no record of the fact that $t$ has reached a barrier, and thus this scenario is no different 
from a scenario where $t$ has not reached the barrier. As a result, we cannot guarantee correctness.

We thus have a proof by contradiction that says that we need to save at least $\kappa=\lceil log(N-M) \rceil$ bits in memory.
Since each state in a system with $\kappa$ threads  can be associated with a number from $1 \ldots \kappa$ let us treat
each state as a count. We can thus say that it is necessary to maintain a count of the number of threads that have
reached the barrier. 
\end{proof}

\begin{table*}[t]
\normalsize
\begin{center}
\begin{tabular}{||p{12cm}|l||}
\hline
\hline
\textbf{Result} & \textbf{Theorem} \\
\hline
We need to maintain a count to at least $N$. &Theorem~\ref{thm:nm} and Corollary~\ref{corr:nm} \\
\hline
It is necessary for any thread to at least broadcast the barrier id once. & Lemma~\ref{lm:bid} \\
\hline
We expect to reduce the number of messages for getting the status of the barrier's count
      by at least $O(N)$ times with a co-ordinator. & Theorems~\ref{thm:dist} and \ref{thm:central} \\
\hline
\hline
\end{tabular}
\caption{Summary of the main results \label{tab:summary}}
\end{center}
\vspace{-7mm}
\end{table*}

\begin{corollary}
\label{corr:nm}
We need to maintain a count till at least $N-M$. If we want a generic implementation then we need to keep a count till
$N$. 
\end{corollary}

\begin{proof}
As proven in Theorem~\ref{thm:nm} we need at least $N-M$ states. We can number each of the states from 1 to $N-M$, and
treat the numbering as a count. We might not be sure about $M$ in a generic implementation because all the threads
might be swapped out at some point of time and $M$ can be 0. To make the implementation generic, we need to keep a count
till $N$ by expending $\lceil log(N) \rceil$ bits.
\end{proof}

\begin{corollary}
If there are $K$ context switches for threads in a barrier group, we need at least $K$ writes to memory.
\end{corollary}

\begin{proof}
Since the barrier controller cannot maintain state, each thread needs to write the status of barriers to memory. 
\end{proof}

Now that we have established some bounds on the space required, let us look at some results when time constraints
are added. Let us assume a synchronous concurrent algorithm where messages are sent in {\em rounds} (see the book
by Nancy Lynch for a more precise definition of the term, {\em round}). Let us consider the moment that some thread
is aware that the barrier needs to be released. At that point it clearly has information regarding all the other
threads. The total number of messages sent in the system of threads has to be at least $N-1$ because all the threads
other than the thread that releases the barrier need to send a message to at least one other thread to make it aware of its
status. We state this result in the form of a lemma.

\begin{lemma}
We need to exchange at least $N-1$ 1-bit messages to decide when a barrier needs to be released.
\end{lemma}

Let us now look at the size of each message. 

\begin{lemma}
For each barrier controller, it is necessary to send a message with the id of the barrier at least once for each barrier entry/release. 
\label{lm:bid}
\end{lemma}

\begin{proof}
We are assuming the participants are unknown, which is most often the case in modern programming languages. 
Hence, the assumption per se is not unrealistic. A thread thus has to inform other threads about the barrier
group that it belongs to. This is necessary because we can have threads from two or more barrier groups 
simultaneously scheduled on a multicore processor. A thread  needs to identify itself uniquely by specifying
the barrier group that it belongs to. We are assuming that the id of the barrier group, and the id of the barrier
are synonymous terms here.
\end{proof}

Lemma~\ref{lm:bid} proves that it is necessary to affix the barrier id with at least one message sent by each
barrier controller.

Let us now look at the reasons and resources required for having a co-ordinator. We would
typically like to have a co-ordinator such that only one release message is sent, and after a set of threads
are swapped in, we minimize the number of messages. Let us compare the two scenarios, where we have a single co-ordinator
and the protocol is fully distributed (all the nodes in a barrier group are co-ordinators). 

\begin{theorem}
\label{thm:dist}
In the fully distributed scenario (each thread is a co-ordinator), we need to send $(N+K)(2\mathcal{C}-2)$ messages to get an updated
status of the current state of the barrier. Here, $K$ is the number of additional context switches.
\end{theorem}

\begin{proof}
Whenever, a thread is swapped in, it does not know about the state of the barrier. It does now know about the count
of the threads that have reached the barrier. To find this information, it needs to send a message to other
threads in its barrier group. However, given the fact that we have unknown participants in the barrier, it does
not know who the other threads in its barrier group are. As a result each thread needs to send $\mathcal{C}-1$ messages 
to other threads to find out about the status of the barrier. Once, a thread receives such a message it needs
to reply. This will create an additional $\mathcal{C}-1$ messages. It is not possible for us to reduce this number further
because if threads co-ordinate among themselves to elect a leader, we need to send at least one message per thread.
The count of messages still comes to $\mathcal{C}-1$. Thus, we need a minimum of $(2\mathcal{C}-2)$ messages to just get an updated
status of the barrier for each thread.

Now, $N$ threads are swapped in at least once. Thus, the total swap-in events are equal to $N+K$, and the message
complexity is $(N+K)(2\mathcal{C}-2)$. 
\end{proof}

Let us look at this bound if we have a co-ordinator (some kind of a centralized entity that maintains a count). 

\begin{theorem}
\label{thm:central}
If the co-ordinator is swapped out $K'$ times, then we need to send $O(\mathcal{C}K')$ messages to transfer the role of
co-ordinatorship, if we have the requirement
that always there is a co-ordinator if possible. 
\end{theorem}

\begin{proof}
Let us assume a co-ordinator gets swapped out. It needs to first find if any of the other threads in its
barrier group can become a co-ordinator. It needs to send $\mathcal{C}-1$ messages to find the threads in its barrier group.
The threads in its barrier group that are active (scheduled) need to send a reply back such that a new co-ordinator
can be elected. This requires $\mathcal{C}-1$ messages. Finally, the old co-ordinator needs to send a message to the new
co-ordinator. The total number of messages is thus $2\mathcal{C}-1$. 

Let us now assume another situation where the old co-ordinator did not find any thread in its barrier group to be active
(scheduled). In this case, it needs to write its count to memory and exit. Any other thread in the barrier group that
later wakes up can read the old count from memory and restart. We thus need 2 messages here. 

The total message complexity for this operation is $O(\mathcal{C}K')$. 
\end{proof}

We observe from Theorems~\ref{thm:dist} and \ref{thm:central} that if
 we do not have a co-ordinator, the number of messages that we need to send to just get information
about the barrier's status is $(N+K)(2\mathcal{C}-2)$. If we have a co-ordinator this number is $O(\mathcal{C}K')$
(bounded by $(2\mathcal{C}-1)K'$). 
Note that $K$ and $K'$ are not the same. $K$ is the number of context switches for any thread in the barrier group,
and $K'$ is the number of context switches by the co-ordinator, which is expected to be $K/N$.  Thus we expect
the number of messages with a co-ordinator to be $O(\mathcal{C}K/N)$. We thus expect to get at least a $O(N)$ times reduction
in the number of messages.

We summarize the main results in Table~\ref{tab:summary}.

\footnotesize
\def\IEEEbibitemsep{0pt plus .5pt}
\bibliographystyle{IEEEtran}

\end{document}